\documentclass[12pt]{article}

\usepackage[utf8]{inputenc}
\usepackage[english]{babel}
\usepackage{amssymb,amsmath}
\usepackage{amsthm}
\usepackage{float}
\usepackage{setspace,lipsum}
\usepackage{amsfonts}
\usepackage{graphicx,caption}
\usepackage[margin=2.4cm]{geometry}
\usepackage{enumerate}
\usepackage{enumitem}
\usepackage{color}
\usepackage[noadjust]{cite}
\usepackage{tikz}
\usepackage{mathtools}
\usepackage{hyperref}
\usepackage{multirow}
\usepackage{array}
\usepackage[section]{placeins}
\usepackage{xfrac}
\usepackage{xspace}
\usepackage{ifthen}
\usepackage{longtable}

\newtheorem{thm}{Theorem}

\newtheorem{cor}[thm]{Corollary}

\theoremstyle{definition}
\newtheorem{exa}{Example}
\newtheorem{defn}{Definition}

\usepackage{xcolor}

\title{ChemicHull: an online tool for determining extremal chemical graphs of maximum degree at most 3 for any degree-based topological indices}
\author{
    S\'ebastien Bonte\textsuperscript{1},
	Gauvain Devillez\textsuperscript{1},
	Valentin Dusollier\textsuperscript{1},
	\\
	Alain Hertz\textsuperscript{2},
	Hadrien M\'elot\textsuperscript{1},
	David Schindl\textsuperscript{3}	
	\\[3mm]
	\footnotesize \textsuperscript{1} Computer Science Department - Algorithms Lab\\[-2mm]
	\footnotesize University of Mons, Mons, Belgium\\[-2mm]
	\footnotesize Corresponding author: hadrien.melot@umons.ac.be\\[3mm]
	\footnotesize \textsuperscript{2} Department of Mathematics and Industrial
	Engineering\\[-2mm]
	\footnotesize Polytechnique Montr\'eal - Gerad, Montr\'eal, Canada\\[3mm]
	\footnotesize \textsuperscript{3} Haute Ecole de Gestion de Gen\`eve\\[-2mm]
	\footnotesize University of Applied Sciences Western Switzerland, Gen\`eve, Switzerland
}

\date{\today}

\newboolean{use_new_notation}
\setboolean{use_new_notation}{true}

\newcommand{\G}[1]{%
\ifthenelse{\equal{#1}{1}}{\ifthenelse{\boolean{use_new_notation}}{F6}{G1}}{%
\ifthenelse{\equal{#1}{2}}{\ifthenelse{\boolean{use_new_notation}}{F1}{G2}}{%
\ifthenelse{\equal{#1}{3}}{\ifthenelse{\boolean{use_new_notation}}{F2}{G3}}{%
\ifthenelse{\equal{#1}{4}}{\ifthenelse{\boolean{use_new_notation}}{F12}{G4}}{%
\ifthenelse{\equal{#1}{5}}{\ifthenelse{\boolean{use_new_notation}}{F3}{G5}}{%
\ifthenelse{\equal{#1}{6}}{\ifthenelse{\boolean{use_new_notation}}{F13}{G6}}{%
\ifthenelse{\equal{#1}{7}}{\ifthenelse{\boolean{use_new_notation}}{F15}{G7}}{%
\ifthenelse{\equal{#1}{7bis}}{\ifthenelse{\boolean{use_new_notation}}{F16}{G7bis}}{%
\ifthenelse{\equal{#1}{8}}{\ifthenelse{\boolean{use_new_notation}}{F5}{G8}}{%
\ifthenelse{\equal{#1}{9}}{\ifthenelse{\boolean{use_new_notation}}{F4}{G9}}{%
\ifthenelse{\equal{#1}{10}}{\ifthenelse{\boolean{use_new_notation}}{F11}{G10}}{%
\ifthenelse{\equal{#1}{11}}{\ifthenelse{\boolean{use_new_notation}}{F7}{G11}}{%
\ifthenelse{\equal{#1}{12}}{\ifthenelse{\boolean{use_new_notation}}{F10}{G12}}{%
\ifthenelse{\equal{#1}{13}}{\ifthenelse{\boolean{use_new_notation}}{F9}{G13}}{%
\ifthenelse{\equal{#1}{14}}{\ifthenelse{\boolean{use_new_notation}}{F8}{G14}}{%
\ifthenelse{\equal{#1}{15}}{\ifthenelse{\boolean{use_new_notation}}{F20}{G15}}{%
\ifthenelse{\equal{#1}{16}}{\ifthenelse{\boolean{use_new_notation}}{F21}{G16}}{%
\ifthenelse{\equal{#1}{17}}{\ifthenelse{\boolean{use_new_notation}}{F17}{G17}}{%
\ifthenelse{\equal{#1}{18}}{\ifthenelse{\boolean{use_new_notation}}{F14}{G18}}{%
\ifthenelse{\equal{#1}{19}}{\ifthenelse{\boolean{use_new_notation}}{F18}{G19}}{%
\ifthenelse{\equal{#1}{20}}{\ifthenelse{\boolean{use_new_notation}}{F19}{G20}}{%
\ifthenelse{\equal{#1}{21}}{\ifthenelse{\boolean{use_new_notation}}{NOTYETDEFINED}{G21}}{%
\textbf{G?}%
}}}}}}}}}}}}}}}}}}}}}}}

\newcommand{\V}[1]{%
\ifthenelse{\equal{#1}{1}}{\ifthenelse{\boolean{use_new_notation}}{\textup{V1}}{\textup{V1}}}{%
\ifthenelse{\equal{#1}{2}}{\ifthenelse{\boolean{use_new_notation}}{\textup{V2}}{\textup{V2}}}{%
\ifthenelse{\equal{#1}{3}}{\ifthenelse{\boolean{use_new_notation}}{\textup{V3}}{\textup{V3}}}{%
\ifthenelse{\equal{#1}{6}}{\ifthenelse{\boolean{use_new_notation}}{\textup{V4}}{\textup{V6}}}{%
\ifthenelse{\equal{#1}{7a}}{\ifthenelse{\boolean{use_new_notation}}{\textup{V5}}{\textup{V7a}}}{%
\ifthenelse{\equal{#1}{7b}}{\ifthenelse{\boolean{use_new_notation}}{\textup{V6}}{\textup{V7b}}}{%
\ifthenelse{\equal{#1}{7c}}{\ifthenelse{\boolean{use_new_notation}}{\textup{V7}}{\textup{V7c}}}{%
\ifthenelse{\equal{#1}{8a}}{\ifthenelse{\boolean{use_new_notation}}{\textup{V8SUPP}}{\textup{V8a}}}{%
\ifthenelse{\equal{#1}{8b}}{\ifthenelse{\boolean{use_new_notation}}{\textup{V9SUPP}}{\textup{V8b}}}{%
\ifthenelse{\equal{#1}{8c}}{\ifthenelse{\boolean{use_new_notation}}{\textup{V8}}{\textup{V8c}}}{%
\ifthenelse{\equal{#1}{8d}}{\ifthenelse{\boolean{use_new_notation}}{\textup{V9}}{\textup{V8d}}}{%
\ifthenelse{\equal{#1}{9a}}{\ifthenelse{\boolean{use_new_notation}}{\textup{V10}}{\textup{V9a}}}{%
\ifthenelse{\equal{#1}{9b}}{\ifthenelse{\boolean{use_new_notation}}{\textup{V11}}{\textup{V9b}}}{%
\ifthenelse{\equal{#1}{9c}}{\ifthenelse{\boolean{use_new_notation}}{\textup{V12}}{\textup{V9c}}}{%
\ifthenelse{\equal{#1}{10a}}{\ifthenelse{\boolean{use_new_notation}}{\textup{V13}}{\textup{V10a}}}{%
\ifthenelse{\equal{#1}{10b}}{\ifthenelse{\boolean{use_new_notation}}{\textup{V14}}{\textup{V10b}}}{%
\ifthenelse{\equal{#1}{10c}}{\ifthenelse{\boolean{use_new_notation}}{\textup{V15}}{\textup{V10c}}}{%
\ifthenelse{\equal{#1}{11a}}{\ifthenelse{\boolean{use_new_notation}}{\textup{V16}}{\textup{V11a}}}{%
\ifthenelse{\equal{#1}{11b}}{\ifthenelse{\boolean{use_new_notation}}{\textup{V17}}{\textup{V11b}}}{%
\ifthenelse{\equal{#1}{11c}}{\ifthenelse{\boolean{use_new_notation}}{\textup{V18}}{\textup{V11c}}}{%
\ifthenelse{\equal{#1}{12a}}{\ifthenelse{\boolean{use_new_notation}}{\textup{V19}}{\textup{V12a}}}{%
\ifthenelse{\equal{#1}{12b}}{\ifthenelse{\boolean{use_new_notation}}{\textup{V20}}{\textup{V12b}}}{%
\ifthenelse{\equal{#1}{12c}}{\ifthenelse{\boolean{use_new_notation}}{\textup{V21}}{\textup{V12c}}}{%
\textbf{V?}%
}}}}}}}}}}}}}}}}}}}}}}}}

\newcommand{\T}[1]{\ensuremath{{\sf T}_{#1}}}
\newcommand{\Path}[1]{\ensuremath{{\sf P}_{#1}}}
\newcommand{\C}[1]{\ensuremath{{\sf C}_{#1}}}
\newcommand{\R}[1]{\mathbb{R}^{#1}}
\newcommand{\conv}{\mathrm{conv}}
\newcommand{\Gc}[1]{\mathcal{G}_{#1}}

\begin{document}

\maketitle

\hrule

\vspace*{0.3cm}

\small
\noindent
\textbf{Abstract.} 

Topological indices are graph-theoretic descriptors that play a crucial role in mathematical chemistry, capturing the structural characteristics of molecules and enabling the prediction of their physicochemical properties.
A widely studied category of topological indices, known as \textit{degree-based topological indices}, are calculated as the sum of the weights of a graph’s edges, where each edge weight is determined by a formula that depends solely on the degrees of its endpoints.

This work focuses exclusively on chemical graphs in which no vertex has a degree greater than 
$3$, a model for conjugated systems. Within a polyhedral framework, each chemical graph is mapped to a point in a three-dimensional space, 
enabling extremal values of any degree-based topological index to be determined through linear optimization over the corresponding polyhedron. Analysis within this framework reveals that extremality is limited to a small subset of chemical graph families, implying that certain chemical graphs can never attain extremality for any degree-based topological index.

The main objective of this paper is to present \emph{ChemicHull}, an online tool we have developed to determine and display extremal chemical graphs for arbitrary degree-based topological indices. To illustrate the power of this tool, we easily recover established results, emphasizing its effectiveness for chemically significant graph classes such as chemical trees and unicyclic chemical graphs. This tool also enabled the identification of a counterexample to a previously published extremal result concerning the Randić index.

\vspace*{0.2cm}

\vspace*{0.2cm}
\hrule

\normalsize

\section{Introduction} \label{sec:intro}

Graph-theoretic descriptors, referred to as topological indices, are useful in mathematical chemistry for capturing molecular structures and predicting their physicochemical behavior~\cite{Codding98, Devillers00, Devillers96, Hansch95, Karelson00, Kier76, Kier86}. As the name implies, these indices depend solely on the structure of the underlying chemical graph, in which vertices correspond to atoms and edges correspond to bonds. 
A key class of these topological indices is the \textit{degree-based topological indices}, defined as the sum of edge weights, with each weight depending solely on the degrees of the edge’s endpoints.
A prominent example is the Randi\'c  index~\cite{Randic1975}, first introduced in 1975. As a reminder, the Randi\'c  index of a graph $G=(V,E)$, denoted $R(G)$, is defined as 
$$
R(G) = \sum_{{uv}\in E} \frac{1}{\sqrt{d(u) d(v)}},
$$
where $d(u)$ denotes the degree of  vertex $u\in V$. 
Over the years, this index has inspired the development of numerous related descriptors. Investigating the extremal values of these indices (i.e., identifying their possible maxima and minima) has generated extensive literature, both from a theoretical standpoint and in connection with practical chemical applications~\cite{das2011abc,zhang2016abc,AFRG22,henning2007albertson,hansen2005albertson,cui2021ag,carballosa2022ag,VPVFS21,che2016forgotten,yan2010ga,deng2018ga,zhong2012harmonic, DENG2013,elumalai2018, das2016,sedlar2015, falahati2017,ali2023maximum,cruz2021sombor,li2022extremal,cruz2021extremal,liu2021reduced,deng2021molecular,ali2019extremal,gutman2004first,nikolic2003zagreb,furtula2010augmented, ali2021augmented,li2008survey, swartz2022survey, hansen2009variable,GFE14,liu2020some, chu2020extremal,raza2020bounds,ali2020sddi,ghorbani2021sddi,hertz2025}.

As pointed out by Patrick Fowler~\cite{PF} and discussed in~\cite{article33},
two distinct definitions of chemical graphs are found in the literature. According to one, they represent the skeletal structures of saturated hydrocarbons, leading to graphs of maximum degree at most 4. 
The other considers unsaturated conjugated systems such as alkenes, polyenes, benzenoids~\cite{kwun_m-polynomials_2018}, and fullerenes~\cite{shigehalli_computation_2016, sharma_degree-based_2022} in which the maximum degree does not exceed three.
Throughout this work, we follow the latter definition and therefore restrict our attention to chemical graphs with a maximum degree of at most three.

Our approach relies on a polyhedral representation of chemical graphs. Specifically, we associate each chemical graph with a point in a three-dimensional space defined by the numbers of edges classified according to the degrees of their endpoints. A degree-based topological index then corresponds to a linear function over this space. This polyhedral framework leads to two major consequences. First, the extremal values of any degree-based topological index can be determined by optimizing a linear function over the associated polyhedron, a straightforward task once the polyhedral representation is known. Second, the structure of the polyhedron reveals that only a small number of families of chemical graphs can be extremal, regardless of the chosen index. This observation provides a unified explanation for the recurrent appearance of the same few families as extremal graphs across different indices and implies that certain chemical graphs can never be extremal for any degree-based topological index.

The structure of the paper is as follows. In Section~\ref{sec:basic}, we outline the relevant background from chemical graph theory  and geometry, and introduce the proposed polyhedral methodology. Section~\ref{sec:extr_points} characterizes all families of extremal chemical graphs, irrespective of the particular degree-based topological index under study. In Section~\ref{sec:chemichull}, we introduce \textit{ChemicHull}~\cite{chemichull}, an online tool that enables interactive exploration of extremal values of degree-based topological indices. Section~\ref{sec:application} presents examples demonstrating how our approach can be used to recover known results from the literature in a simple and systematic way. Finally, Section~\ref{sec:futurework} concludes the paper with a summary and a discussion of directions for future research.

Throughout the paper, we have taken particular care to ensure that the exposition is both accessible and self-contained. Numerous examples and reminders of key concepts are provided, with special attention given to graph classes of central importance in chemical graph theory, such as trees and unicyclic graphs. Our goal is to make the proposed framework transparent and readily applicable to researchers studying extremal problems on degree-based topological indices.

\section{Basics and notations}\label{sec:basic}

Let $G = (V, E)$ be a simple undirected graph of {\em order} $n = |V|$ and {\em size} $m = |E|$. The \emph{degree} $d(v)$ of a vertex $v \in V$ is the number of vertices adjacent to $v$. The maximum degree of $G$ is denoted by $\Delta(G)$. In what follows, $\Path{n}$ and $\C{n}$ denote the path of order $n$ and the cycle of order $n$, respectively. Also, $\T{n}$ is the tricyclic graph obtained by adding two adjacent vertices linked to the endpoints of a path $\Path{n-2}$. For illustration, $\Path{5}$, $\C{6}$ and $\ensuremath{{\sf T}}_7$ are shown in Figure \ref{fig:PCT}.
\begin{figure}[!htb]
	\centering    \includegraphics[scale=1.0]{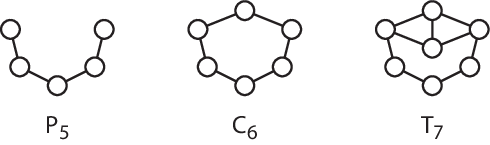}
	\vspace{-0.2cm}\caption{Examples of graphs $\Path{n}$, $\C{n}$ and $\ensuremath{{\sf T}}_n$.}
	\label{fig:PCT}
\end{figure}

\subsection{Chemical graphs}

A chemical graph is a simple connected  undirected graph in which vertices represent atoms and edges represent the chemical bonds between them. The set of chemical graphs of order $n$, size $m$, and maximum degree at most $\Delta$ is denoted by $\Gc{\Delta}(n,m)$. For instance, $\Gc{4}(n,n-1)$ represents the set of chemical trees with maximum degree not exceeding 4.
Throughout this work, we restrict our attention to chemical graphs with maximum degree at most 3. Accordingly, unless explicitly noted, any chemical graph of order 
$n$ and size 
$m$ will be understood to belong to $\Gc{3}(n,m)$. 

For a graph $G$ and  integers $i,j$ with $1\le i\le j$, an $ij$-edge in $G$ is an edge whose endpoints have degree $i$ and $j$. We denote by $m_{ij}$ the number of $ij$-edges in $G$, and by $n_i$ the number of vertices of degree $i$ in $G$. 
The only chemical graphs of order 1 and 2 are the paths 
$\Path{1}$ and $\Path{2}$, respectively.
Accordingly, to exclude trivial instances, we restrict our attention to chemical graphs with a minimum of three vertices, which gives
\begin{align} 2 \le n - 1 \le m \le \min\left\{\left\lfloor \frac{3n}{2} \right\rfloor, \frac{n(n-1)}{2}\right\}.\label{val_n_m}
\end{align}

Since chemical graphs are connected, Inequalities (\ref{val_n_m}) imply $m_{11}=0$. Hence, a chemical graph can have at most 5 nonzero $m_{ij}$ values, namely $m_{12}, m_{13}, m_{22}, m_{23}$ and $m_{33}$. One can readily verify that
\begin{align}
 \label{eq_n1}  n_1 & = m_{12} + m_{13}\\[-3pt]
 \label{eq_n2}  n_2 & = \frac{m_{12} + 2 m_{22} + m_{23} }{2}\\[-3pt] 
 \label{eq_n3}  n_3 & = \frac{m_{13} + m_{23} + 2 m_{33} }{3}\\[-3pt]
\label{n}n&=n_1+n_2+n_3=\frac{3}{2} m_{12} + \frac{4}{3} m_{13} + m_{22} + \frac{5}{6} m_{23} + \frac{2}{3} m_{33}\\[-3pt]
\label{m}m&=m_{12}+m_{13}+m_{22}+m_{23}+m_{33}.
\end{align}

As emphasized in the introduction, our goal is to study the extremal properties of degree-based topological indices of chemical graphs. These indices are computed as the sum of edge weights, where the weight of an edge 
$vw$ is determined by a formula involving the degrees of 
$v$ and $w$.

\begin{defn}\label{def:degreetopo}
A \emph{degree-based topological index} $I$ for a chemical graph $G \in \Gc{3}(n, m)$ with $n \ge 3$ is a function of the form
\vspace{-0.3cm}$$I(G)  = c_{12}m_{12}+c_{13}m_{13}+c_{22}m_{22}+c_{23}m_{23}+c_{33}m_{33},$$
where every $c_{ij}$ is a real number, and $m_{ij}$ is the number of $ij$-edges in $G$.
\end{defn}

For example, the Randi\'c index (see Section~\ref{sec:intro}) is the degree-based topological index with $c_{ij} = \frac{1}{\sqrt{i j}}$. 
Note that given $n$, $m$ and three of the five $m_{ij}$ values, the remaining two can be determined. For example, if $m_{12}$, $m_{13}$ and $m_{33}$ are known, it follows from Equations (\ref{n}) and (\ref{m}) that
\begin{align}
\label{get_x22}
m_{22} &= 6n - 5m - 4 m_{12} - 3 m_{13} + m_{33}\\[-3pt]
\label{get_x23}
m_{23} &= 6m - 6n + 3 m_{12} + 2 m_{13} - 2 m_{33}.
\end{align}

In this paper, for a pair $(n,m)$ satisfying~\eqref{val_n_m}, we consider tuples $(m_{12}, m_{13}, m_{33})$ for which at least one corresponding chemical graph exists.

\begin{defn} \label{def:realp}
	A point $(m_{12}, m_{13}, m_{33})$ is said to be \emph{realizable} for a pair $(n, m)$ if there exists a graph $G \in \Gc{3}(n, m)$ of order $n$ and size $m$ containing $m_{12}$ 12-edges, $m_{13}$ 13-edges and $m_{33}$ 33-edges.
\end{defn}

\subsection{Basic  notions in geometry}
\label{sec:geometry}

To study the extremal properties of chemical graphs of order 
$n$ and size 
$m$, we consider the convex hull of all realizable points for a given pair 
$(n,m)$. We briefly recall the definition of a convex hull along with two equivalent representations.

\begin{defn}
Let $S = \{x_1, \dots, x_p\}$ be a finite set of $p$ points in $\mathbb{R}^k$. The \emph{convex hull} of $S$, denoted $\conv(S)$, is the set of all convex combinations of the points in $S$; explicitly, 
\[
  \conv(S) = \left\{ \sum_{i=1}^p \lambda_i x_i \;\middle|\; 
  \lambda_i \geq 0,\; \sum_{i=1}^p \lambda_i = 1 \right\}.
 \]

An \emph{extreme point} of a polyhedron is a point that cannot be expressed as the inner point of any line segment within the polyhedron. Since $S$ is finite, its convex hull forms a \emph{polytope} $P$ which can be expressed in the following two equivalent ways:
\begin{enumerate}
  \item The \emph{V-representation} of $P$ is the subset $S'\subseteq S$ consisting of its extreme points,  meaning that $P$ is the convex hull of $S'$; that is, $\mathrm{conv}(S) =P=
  \mathrm{conv}(S')$.
  \item 
  The \emph{H-representation} of $P$ is  the set of linear inequalities corresponding to the intersection of halfspaces. Explicitely,
  $
  \mathrm{conv}(S) =P= \{ x \in \R{k} \mid Ax \leq b \}$
  for some matrix $A$ and vector $b$.  
  This system is \emph{minimal} in the sense that removing any inequality strictly enlarges the corresponding region in $\R{k}$.
\end{enumerate}
\end{defn}

Thus, the polytope $P=\mathrm{conv}(S)$ can be seen either as the smallest convex set containing its points or as the intersection of all half-spaces that include 
$S$. A \emph{facet-defining inequality} specifies a \emph{facet} of 
$P$, i.e., a face of maximal dimension 
$k-1$. Some valid inequalities, however, may be redundant: they hold for all points of 
$P$ without defining new facets. Facet-defining inequalities thus provide a minimal and exact description of the polytope via its bounding hyperplanes

\begin{exa} \label{exa:convex}
Let $S = \{(0, 0), (1, 0), (1, 1), (1, 2), (2, 0), (2, 1), (3, 0), (3, 1)\}$ be a set of points in $\R{2}$. The gray polytope in Figure~\ref{fig:conv} corresponds to $\conv(S)$, its V-representation is \{(0, 0), (1, 2), (3, 0), (3, 1)\} and its H-representation consists of the following four inequalities:
\begin{align}
x_1 & \le  3,\\
- x_2 & \le 0,\\
x_1 + 2 x_2 & \le 5,\\
x_1 - 2 x_2 & \le 0.
\end{align}

The inequality $-x_1 + 2 x_2 \le 4$ is valid  for all points of $S$ but is redundant (the inequality defines a half-plane lying below the dotted line in Figure~\ref{fig:conv}).
\end{exa}

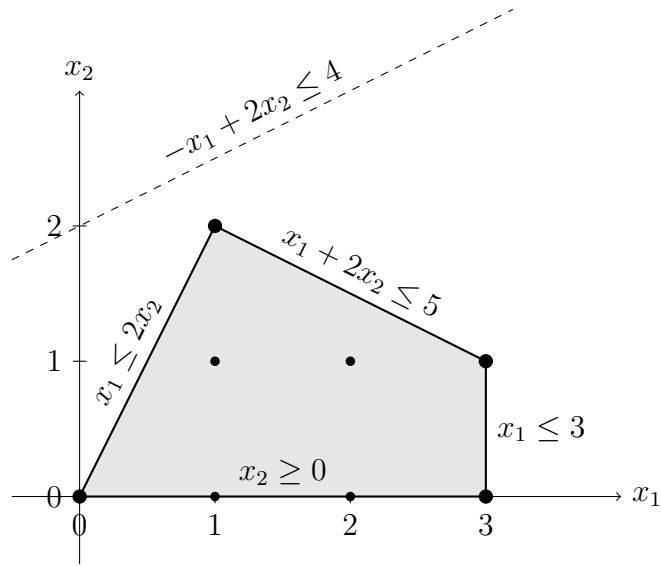
\begin{figure}[h!]
\centering

\begin{tikzpicture}[scale=1.8]

\draw[->] (-0.5,0) -- (4,0) node[right] {$x_1$};
\draw[->] (0,-0.5) -- (0,3) node[above] {$x_2$};

\foreach \x in {0,1,2,3} \draw (\x,0.05) -- (\x,-0.05) node[below] {\x};
\foreach \y in {0,1,2} \draw (0.05,\y) -- (-0.05,\y) node[left] {\y};

\coordinate (A) at (0,0);
\coordinate (B) at (3,0);
\coordinate (C) at (3,1);
\coordinate (D) at (1,2);

\fill[gray!20] (A) -- (B) -- (C) -- (D) -- cycle;

\draw[thick] (A) -- (B) node[midway, above] {$x_2 \ge 0$};
\draw[thick] (B) -- (C) node[midway, right] {$x_1 \le 3$};
\draw[thick] (C) -- (D) node[midway, sloped, above] {$x_1 + 2 x_2 \le 5$};
\draw[thick] (D) -- (A) node[midway, sloped, above] {$x_1 \le 2 x_2 $ };

\draw[dashed, -] (-0.5,1.75) -- (3.2,3.6) node[midway, sloped, above]{$-x_1 + 2 x_2 \le 4$};

\foreach \v in {A,B,C,D} \fill (\v) circle (1.5pt);

\fill[black] (1,1) circle (1pt);
\fill[black] (1,2) circle (1pt);
\fill[black] (2,1) circle (1pt);
\fill[black] (1,0) circle (1pt);
\fill[black] (2,0) circle (1pt);

\end{tikzpicture}

\caption{A 2D polytope illustrating facet-defining inequalities and a valid (but redundant) inequality (dashed line).}\label{fig:conv}
\end{figure}

Let 
$S$ denote the set of all feasible solutions of an integer optimization problem. A point is an \emph{interior point} of the polyhedron $\conv(S)$ if  it strictly satisfies every inequality in the H-representation
of $\conv(S)$. Standard results from integer programming (see, for example \cite{Wolsey2020}) show that a linear objective function never attains its optimum over $S$ at an interior point of $\conv(S)$. Instead, its optimum over 
$S$ is achieved at least at one extreme point of 
$\conv(S)$. Hence, knowing the V-representation of 
$\conv(S)$ is sufficient to identify at least one optimal solution.
 
 As an illustration, in Example \ref{exa:convex}, if 
$f(x_1, x_2) = x_2$, the maximum occurs only at $(1, 2)$, while the minimum is reached at all points on the facet $x_2 \ge 0$
, namely $\{(0, 0), (1, 0), (2, 0), (3, 0)\}$ with $(0, 0)$ and $(3, 0)$ being extreme points. The points 
$(1,1)$ and $(2,1)$
 are interior to 
$\conv(S)$ and do not yield an optimum over 
$S$ for any linear objective function.

\subsection{Polytopes of chemical graphs}

Definition \ref{def:realp} maps each chemical graph in $\Gc{3}(n,m)$ to a point in 
$\R{3}$. We now define the polytope  induced by all points in 
$\R{3}$
 that are realizable for a given pair 
$(n,m)$.
\begin{defn}\label{def:polytope}
	Let $n$ and $m$ be two integers satisfying~\eqref{val_n_m}. The associated polytope $\mathcal{P}_{n, m}$ is the convex hull of all points $(m_{12},m_{13},m_{33})$ that are realizable  for $(n,m)$.  
\end{defn}

According to this definition, we consider a \emph{full-dimensional} polytope $\mathcal{P}_{n, m}$ to have dimension 3. Note that a single point in $\mathcal{P}_{n, m}$ may correspond to multiple chemical graphs. For instance, the point $(0,4,4)$ in $\mathcal{P}_{8,8}$ corresponds, by Definitions~\ref{def:realp} and \ref{def:polytope} and Equations~\eqref{get_x22} and~\eqref{get_x23}, to all chemical graphs of order $n = 8$, size $m=8$, and with $m_{12}=0$, $m_{13}=4$, $m_{22}=m_{23}=0$ and $m_{33}=4$. There are two non-isomorphic chemical graphs\footnote{The list of non-isomorphic graphs with these values is exhaustive, as computed with the help of the discovery system PHOEG~\cite{phoeg}.} corresponding to these values, shown in Figure~\ref{fig:twographs}. As another example, for any even $n$, the point $(0,0,m=\frac{3n}{2})$
corresponds to all cubic graphs of order $n$.

\begin{figure}[!htb]

\centering\includegraphics[scale=1.0]{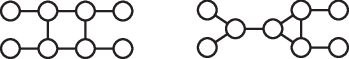}
	\caption{Two non-isomorphic chemical graphs with coordinates $(0,4,4)$ in $\mathcal{P}_{8,8}$.}
	\label{fig:twographs}
\end{figure}

Although a given point may correspond to more than one chemical graph, the knowledge of 
$\mathcal{P}_{n,m}$ is sufficient to bound the invariants that can be expressed as linear functions of $m_{12}$, $m_{13}$ and $m_{33}$. As noted earlier, it is enough to focus on the vertices or the facets of the polytope. This applies in particular to degree-based topological indices. Indeed, for a degree-based topological index 
$I$ defined by coefficients 
$c_{ij}$, we have, by Definition~\ref{def:degreetopo} and Equations~\eqref{get_x22} and~\eqref{get_x23}:
\begin{align*}
    I(G)=&c_{12}m_{12}+c_{13}m_{13}+c_{22}m_{22}+c_{23}m_{23}+c_{33}m_{33}\\
    =&c_{12}m_{12}+c_{13}m_{13}+c_{22}(6n - 5m - 4 m_{12} - 3 m_{13} + m_{33})\\
    &+c_{23}(6m - 6n + 3 m_{12} + 2 m_{13} - 2 m_{33}) + c_{33} m_{33}\\
    =&(c_{12}-4c_{22}+3c_{23})m_{12}+(c_{13}-3c_{22}+2c_{23})m_{13}+(c_{22}-2c_{23}+c_{33})m_{33}\\
    &+(6n-5m)c_{22}+(6m-6n)c_{23}.
\end{align*}

\begin{defn} 
Let $I$ be a degree-based topological index defined by $c_{ij}$ values. We define
\begin{itemize}[nosep]
\item $\widehat{I}(m_{12}, m_{13}, m_{33}) = c'_{12} m_{12} + c'_{13} m_{13} + c'_{33} m_{33}
$, where
\begin{itemize}\setlength\itemsep{1pt}
\item[$\circ$] $c'_{12} = c_{12}-4c_{22}+3c_{23}$, 
\item[$\circ$] $
c'_{13} = c_{13}-3c_{22}+2c_{23}$,
\item[$\circ$] $
c'_{33} = c_{22}-2c_{23}+c_{33},$
\end{itemize}
\item $
C_{I}(n,m) = (6n-5m)c_{22}+(6m-6n)c_{23}.
$
\end{itemize}
\end{defn}
Note that for a chemical graph $G$ of order $n$, size $m$, and with $m_{ij}$ $ij$-edges, we have $I(G) = \widehat{I}(m_{12}, m_{13}, m_{33}) + C_{I}(n,m)$.
Hence, maximizing or minimizing 
$I(G)$ over all graphs of order 
$n$ and size 
$m$ is equivalent to maximizing or minimizing the linear function 
$\widehat{I}(m_{12}, m_{13}, m_{33})$. The optimal value of 
$I(G)$ is then obtained by adding 
$C_I(n, m)$ to the previously determined optimum.

\begin{exa} \label{exa:maxRandic}
Maximizing or minimizing the Randi\'c index 
$R$ for chemical graphs of order 
$n$ and size 
$m$ reduces to optimizing the linear function
\begin{align*}
\widehat{R}(m_{12}, m_{13}, m_{33}) & = \frac{\sqrt{2} + \sqrt{6} - 4}{2}m_{12} + \frac{2 \sqrt{3} + 2 \sqrt{6} - 9}{6}m_{13} + \frac{5 - 2 \sqrt{6}}{6}m_{33}\\
& \simeq -0.068 m_{12} - 0.106 m_{13} + 0.016 m_{33}.
\end{align*}
Once the optimal values   $m^*_{12}, m^*_{13}$ and $m^*_{33}$ are determined, one can simply add 
\begin{align*}
C_R(n, m) &= \frac{6n-5m}{2} + \frac{6m-6n}{\sqrt{6}},\\
&= (3 - \sqrt{6}) n + \left(\sqrt{6} - \frac{5}{2}\right) m,\\
&\simeq 0.55 n - 0.05 m,
\end{align*}
to $\widehat{R}(m^*_{12}, m^*_{13}, m^*_{33})$ to get the maximum (resp. minimum) value of $R$.
\end{exa}

By abuse of notation, we denote by 
$I(p)$ the value of 
$I(G)$ for any graph 
$G$ associated with a realizable point $
p=(m_{12},m_{13},m_{33})$.

As indicated in Section~\ref{sec:basic}, since $\widehat{I}$ is linear, it is sufficient to know the V-representation of $\mathcal{P}_{n, m}$ in order to maximize or minimize $\widehat{I}$. That is what is done in the following section.

\section{Extreme points} \label{sec:extr_points}

Our recent paper ~\cite{delta3} provides a complete description of all possible polytopes 
$\mathcal{P}_{n, m}$ for any pair 
$(n,m)$ satisfying \eqref{val_n_m}. In this section, we present the principal results of that study from the viewpoint of their applicability to chemistry, particularly for researchers interested in determining the chemical graphs that maximize or minimize a given degree-based topological index. Accordingly, we concentrate on the identification of the extreme points of these polytopes for a fixed pair $(n,m)$.
Remarkably, the number of such points remains very small, even when $n$ and $m$ are arbitrarily large, never exceeding 16. 

Let's first focus on the polytopes $\mathcal{P}_{n, m}$ corresponding to the pairs $(n,m)$ that satisfy the following condition:
\begin{align}\max\{12,n-1\}\le m\leq\left\lfloor\frac{3n-3}{2} \right\rfloor. \label{m>12}\end{align} As proved in~\cite{delta3}, for any such polytope $\mathcal{P}_{n, m}$, its set of extreme points is a subset of the 21 points whose coordinates are given in Table~\ref{tab:vi}. Table~\ref{tab:extr_points1} lists the extreme points of these polytopes all of which are full-dimensional. 
Note that, in all cases shown in Table~\ref{tab:extr_points1}, the extreme points are identical for both even and odd values of 
$n$; the only exception is the point 
$\V{8d}$, which appears exclusively when 
$n$ is odd. 

The extreme points of the remaining polytopes that do not satisfy condition~\eqref{m>12} (i.e., the polytopes for chemical graphs with at least two but fewer than twelve edges, or with more than $\left\lfloor\frac{3n-3}{2} \right\rfloor$ edges) are given in Table~\ref{tab:extr_points2}. 

\begin{table}[h!]
\renewcommand{\arraystretch}{1.22}	\centering
	\caption{Twenty-one  points $(m_{12},m_{13},m_{33})$}\label{tab:vi}~\\
	\begin{tabular}{|c|c|c|c|}\hline		
		Id&$m_{12}$&$m_{13}$&$m_{33}$\\\hline	
		\V{1} &$0$&$0$&$0$\\	
		\V{2}&$2$&$0$&$0$\\
		\V{3}&$0$&$0$&$1$ \\
		\V{6}&$0$&$0$&$5m - 6n$ \\
		\V{7a}&$\frac{ 6n -5m - 3((m-2n) \bmod 4)}{4}$&$(m-2n) \bmod 4$&$ 0$\\
		\V{7b}&$\frac{6n-5m + (m-2n) \bmod 4}{4}$&$0$&$(m-2n) \bmod 4$\\
		\V{7c}&$\frac{6n-5m - (2n-m) \bmod 4}{4}$&$0$&$0$\\
		\V{8c}&$0$&$\frac{ 3n-2m - n \bmod 2}{2}$&$\frac{4m - 3n - 3(n \bmod 2)}{2}$\\
		\V{8d}&$1$&$\frac{3n-2m - 3}{2}$&$\frac{4m - 3n - 1}{2}$\\
		\V{9a}&$3m - 3n - 2$&$3m - 3n - 2$&$6m - 6n - 1$\\
		\V{9b}&$3m - 3n - 1$&$0$&$6m - 6n - 1$\\
		\V{9c}&$0$&$3m - 3n - 2$&$6m - 6n - 3$\\
		\V{10a}&$0$&$\frac{6n-5m - m \bmod 3}{3}$&$0$ \\
		\V{10b}&$m \bmod 3$&$\frac{6n-5m - 4(m \bmod 3)}{3}$&$0$ \\
		\V{10c}&$0$&$\frac{6n-5m + (2m) \bmod 3}{3}$&$(2m) \bmod 3$\\
		\V{11a}&$\frac{3n-2m - m \bmod 3}{3}$&$0$&$\frac{7m - 6n - 4(m \bmod 3)}{3}$\\
		\V{11b}&$\frac{3n-2m - 2((2m) \bmod 3)}{3}$&$(2m) \bmod 3$&$\frac{7m - 6n + (2m) \bmod 3}{3}$\\
		\V{11c}&$\frac{ 3n-2m - m \bmod 3}{3}$&$0$&$\frac{7m - 6n - m \bmod 3}{3}$\\
		\V{12a}&$0$&$3n-3m + 1$&$0$\\
		\V{12b}&$0$&$0$&$3m - 3n - 1$\\
		\V{12c}&$1$&$0$&$3m - 3n + 1$\\\hline
	\end{tabular}
\end{table}

\begin{table}
	\scriptsize	\centering
	\caption{Extreme points for $\mathcal{P}_{n, m}$ when the pair $(n,m)$ satisfies condition~\eqref{m>12}}\label{tab:extr_points1}~\\
	\setlength{\tabcolsep}{2pt}
	\begin{tabular}{ccccccccccccccccccccc|cc}
\multicolumn{23}{c}{$m=n -1$}\\\cline{1-22}
\V{1}&\V{2}&\V{3}&\V{6}&\V{7a}&\V{7b}&\V{7c}&\V{8c}&\V{8d}&\V{9a}&\V{9b}&\V{9c}&\V{10a}&\V{10b}&\V{10c}&\V{11a}&\V{11b}&\V{11c}&\V{12a}&\V{12b}&\V{12c}&$n\bmod 2$&\\
&x&&&x&x&x&x&&&&&x&x&x&x&x&x&x&&&0&\\
&x&&&x&x&x&x&x&&&&x&x&x&x&x&x&x&&&1&\\\cline{1-22}
\multicolumn{23}{c}{\vspace{-2pt}}\\
\multicolumn{23}{c}{$m=n$}\\\cline{1-22}
\V{1}&\V{2}&\V{3}&\V{6}&\V{7a}&\V{7b}&\V{7c}&\V{8c}&\V{8d}&\V{9a}&\V{9b}&\V{9c}&\V{10a}&\V{10b}&\V{10c}&\V{11a}&\V{11b}&\V{11c}&\V{12a}&\V{12b}&\V{12c}&$n\bmod 2$&\\
x&&&&x&x&x&x&&&&&x&x&x&x&x&x&&&&0&\\
x&&&&x&x&x&x&x&&&&x&x&x&x&x&x&&&&1&\\\cline{1-22}

\multicolumn{23}{c}{\vspace{-2pt}}\\
\multicolumn{23}{c}{$m=n+1$}\\\cline{1-22}
\V{1}&\V{2}&\V{3}&\V{6}&\V{7a}&\V{7b}&\V{7c}&\V{8c}&\V{8d}&\V{9a}&\V{9b}&\V{9c}&\V{10a}&\V{10b}&\V{10c}&\V{11a}&\V{11b}&\V{11c}&\V{12a}&\V{12b}&\V{12c}&$n\bmod 2$&\\
x&&x&&x&x&x&x&&x&x&x&x&x&x&x&x&x&&&&0&\\
x&&x&&x&x&x&x&x&x&x&x&x&x&x&x&x&x&&&&1&\\\cline{1-22}

\multicolumn{23}{c}{\vspace{-2pt}}\\
\multicolumn{23}{c}{$n + 1<m<\frac{6n}{5}$}\\
\V{1}&\V{2}&\V{3}&\V{6}&\V{7a}&\V{7b}&\V{7c}&\V{8c}&\V{8d}&\V{9a}&\V{9b}&\V{9c}&\V{10a}&\V{10b}&\V{10c}&\V{11a}&\V{11b}&\V{11c}&\V{12a}&\V{12b}&\V{12c}&$n\bmod 2$&$\;6n{-}5m{\in}\{1,2,5\}$ \\ \hline
x&&&&x&x&x&x&&&&&x&x&x&x&x&x&&x&x&0&\\
x&&&&&x&x&x&&&&&x&&x&x&x&x&&x&x&0&\checkmark\\
x&&&&x&x&x&x&x&&&&x&x&x&x&x&x&&x&x&1&\\
x&&&&&x&x&x&x&&&&x&&x&x&x&x&&x&x&1&\checkmark\\
\hline
\multicolumn{23}{c}{\vspace{-2pt}}\\
\multicolumn{23}{c}{$\frac{6n}{5} \le m\leq \left\lfloor\frac{3n-3}{2} \right\rfloor$}\\\cline{1-22}
\V{1}&\V{2}&\V{3}&\V{6}&\V{7a}&\V{7b}&\V{7c}&\V{8c}&\V{8d}&\V{9a}&\V{9b}&\V{9c}&\V{10a}&\V{10b}&\V{10c}&\V{11a}&\V{11b}&\V{11c}&\V{12a}&\V{12b}&\V{12c}&$n\bmod 2$& \\ \cline{1-22}
&&&x&&&&x&&&&&&&&x&x&x&&x&x&0&\\
&&&x&&&&x&x&&&&&&&x&x&x&&x&x&1&\\\cline{1-22}
\end{tabular}

\end{table}

\begin{table}[!htb]
	\vspace{0.5cm}\scriptsize\centering
	\caption{Extreme points for $\mathcal{P}_{n, m}$ when the pair $(n,m)$ does not satisfy condition~\eqref{m>12}}\label{tab:extr_points2}~\\
	\begin{tabular}{c|c|l} 
		\cline{1-3}
		$n$ & $m$ &  \multicolumn{1}{c}{extreme points} \\ \cline{1-3}
		$3$ & $2$ &  (2,0,0)\\ 
		$3$ & $3$ &  (0,0,0)\\ 
		$4$ & $3$ &  (0,3,0), (2,0,0)\\ 
		$4$ & $4$ &  (0,0,0), (0,1,0)\\ 
		$4$ & $5$ &  (0,0,1)\\ 
		$5$ & $4$ &  (1,2,0), (2,0,0)\\ 
		$5$ & $5$ &  (0,0,0), (0,1,0), (0,2,1), (1,0,0)\\ 
		$5$ & $6$ &  (0,0,0), (0,0,1), (0,1,3)\\ 
		$6$ & $5$ &  (0,4,1), (1,2,0), (2,0,0), (2,1,0)\\ 
		$6$ & $6$ &  (0,0,0), (0,2,0), (0,3,3), (1,0,0), (1,1,1)\\ 
		$6$ & $7$ &  (0,0,0), (0,0,1), (0,1,2), (0,2,5), (1,0,3)\\ 
		$7$ & $6$ &  (0,4,0), (1,3,1), (2,0,0), (3,0,0)\\ 
		$7$ & $7$ &  (0,0,0), (0,2,0), (0,3,2), (1,0,0), (1,1,0), (1,2,3), (2,0,1)\\ 
		$7$ & $8$ &  (0,0,0), (0,0,1), (0,1,1), (0,1,3), (0,2,4), (1,0,2), (1,0,3), (1,1,5)\\ 
		$7$ & $9$ &  (0,0,3), (0,0,5), (0,1,6), (1,0,7)\\ 
		$8$ & $7$ &  (0,4,0), (0,5,2), (1,3,0), (2,0,0), (2,2,1), (3,0,0)\\ 
		$8$ & $8$ &  (0,0,0), (0,2,0), (0,3,1), (0,4,4), (2,0,0), (2,0,1), (2,1,3)\\ 
		$8$ & $9$ &  (0,0,0), (0,0,1), (0,1,0), (0,1,3), (0,3,6), (1,0,1), (1,1,5), (2,0,5)\\ 
		$8$ & $10$ &  (0,0,2), (0,0,5), (0,2,8), (1,0,6), (1,0,7)\\ 
		$9$ & $8$ &  (0,4,0), (0,5,1), (1,4,2), (2,0,0), (2,2,0), (3,0,0), (3,1,1)\\ 
		$9$ & $9$ &  (0,0,0), (0,3,0), (0,4,3), (1,3,4), (2,0,0), (3,0,3)\\ 
		$9$ & $10$ &  (0,0,0), (0,0,1), (0,1,0), (0,1,3), (0,2,2), (0,3,5), (1,0,0), (1,1,5), (1,2,6), (2,0,4), (2,0,5)\\ 
		$9$ & $11$ &  (0,0,1), (0,0,5), (0,2,7), (1,0,5), (1,0,7), (1,1,8)\\ 
		$10$ & $9$ &  (0,4,0), (0,5,0), (0,6,3), (2,0,0), (3,0,0), (3,1,0), (4,0,1)\\ 
		$10$ & $10$ &  (0,0,0), (0,3,0), (0,4,2), (0,5,5), (1,2,0), (2,0,0), (2,2,4), (3,0,2), (3,0,3)\\ 
		$10$ & $11$ &  (0,0,0), (0,0,1), (0,1,0), (0,1,3), (0,2,1), (0,4,7), (1,0,0), (1,1,5), (2,0,3), (2,0,5), (2,1,6)\\ 
		$11$ & $10$ &  (0,4,0), (0,5,0), (0,6,2), (1,4,0), (1,5,3), (2,0,0), (3,2,2), (4,0,0), (4,0,1)\\ 
		$11$ & $11$ &  (0,0,0), (0,3,0), (0,4,1), (0,5,4), (1,4,5), (2,0,0), (2,1,0), (3,0,1), (3,0,3), (3,1,4)\\ 
		$12$ & $11$ &  (0,4,0), (0,5,0), (0,6,1), (0,7,4), (2,0,0), (2,3,0), (4,0,0), (4,0,1), (4,1,2)\\ 
		even $\ge 4$ & $\frac{3n}{2}$ &  (0,0,$m$) \\ 
		odd $ \ge 5$ & $\frac{3n-1}{2}$ &  (0,0,$m-$2) \\ 
		even $\ge 6$ & $\frac{3n-2}{2}$ &  (0,0,$m-$4), (0,0,$m-$3), (0,1,$m-$1) \\ \cline{1-3}
	\end{tabular}
\end{table}

\begin{exa} \label{exa:P8_10}
Suppose we are interested in determining the optimal values of a degree-based topological index $I$ for chemical graphs with 8 vertices and 8 edges. As indicated in Table~\ref{tab:extr_points2}, the polytope 
$\mathcal{P}_{8, 8}$ contains only seven extreme points, namely, $(0,0,0)$, $(0,2,0)$, $(0,3,1)$, $(0,4,4)$, $(2,0,0)$, $(2,0,1)$, and $(2,1,3)$.
To determine the optimal values of the index $I$, it is sufficient to evaluate 
$\widehat{I}(p)$ for these seven points 
$p$. For example, if $I$ is the Randi\'c index $R$, the formula in Example~\ref{exa:maxRandic} gives $
\widehat{R}(0, 0, 0)  = 0$, $
\widehat{R}(0, 2, 0)  \simeq -0.2123$, $
\widehat{R}(0, 3, 1)  \simeq -0.3016$, $
\widehat{R}(0, 4, 4)  \simeq -0.3573$, $
\widehat{R}(2, 0, 0)  \simeq -0.1362$, $
\widehat{R}(2, 0, 1)  \simeq -0.1194$,
and $
\widehat{R}(2, 1, 3)  \simeq -0.1919$.

These observations show that the chemical graphs maximizing the Randić index all correspond to the point 
$(0,0,0)$, while those minimizing it correspond to the point 
$(0,4,4)$. One can readily verify that the unique graph with 8 vertices and 8 edges associated with 
$(0,0,0)$ is the cycle 
$\C{8}$. Likewise, the two graphs depicted in Figure~\ref{fig:twographs} are precisely those that achieve the minimum Randić index among all chemical graphs on 8 vertices with 8 edges.
In contrast, the graph shown in Figure~\ref{fig:onegraph} corresponds to the point 
$(1,1,1)$, which neither maximizes nor minimizes the Randić index. Note that 
this point lies in the interior of 
$\mathcal{P}_{8, 8}$, since it strictly satisfies all inequalities in its H-representation (see Figure~\ref{fig:chem1}). Consequently, the graph in Figure~\ref{fig:onegraph} is not extremal in 
$\Gc{3}(8,8)$, not only for the Randić index, but for any degree-based topological index.
\end{exa}
\begin{figure}[!htb]
\centering\includegraphics[scale=1.0]{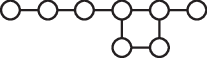}
\caption{A graph in 
$\Gc{3}(8,8)$ that is not extremal with respect to any degree-based topological index.}
	\label{fig:onegraph}
\end{figure}

\begin{exa}[Chemical trees, unicyclic and bicyclic chemical graphs] \label{exa:trees}
Table~\ref{tab:extr_points1} shows that, for any even number 
$n\geq 13$, no matter how large, the polytope 
$\mathcal{P}_{n, n-1}$
 corresponding to chemical trees of order 
$n$ has only 12 extreme points. A thirteenth extreme point $\V{8d}=\left(1, \frac{n-1}{2}, \frac{n-5}{2}\right)$ must be included when 
$n$ is odd. Similarly, the polytopes corresponding to chemical unicyclic ($m=n$) and bicyclic ($m=n+1$) graphs of even order $n$
and size $m$ at least 12
have 11 and 15 extreme points, respectively, with 
$\V{8d}$ added in each case if 
$n$ is odd. Thus, as illustrated in Section \ref{sec:application}, it is sufficient to evaluate only a small number of points to determine a chemical tree, unicyclic, or bicyclic chemical graph that optimizes a degree-based topological index.
\end{exa}

\section{The online tool ChemicHull} \label{sec:chemichull}
As indicated in Section  \ref{sec:geometry}, a polytope can be represented either by its extreme points (V-representation) or by its facets (H-representation). It is worth noting that, although two polytopes may share the same set of extreme points, their facets can differ. For instance, as observed in Section~\ref{sec:extr_points}, all chemical trees with an even number $n\geq 13$ of vertices share the same set of extreme points. However, depending on the value of 
$n \bmod 12$, their facets may vary. For instance, when 
$n \bmod 12=2$, the polytope 
$\mathcal{P}_{n,n-1}$  has ten facets, whereas when 
$n \bmod 12=4$, it has only seven. This difference arises because several extreme points coincide in the latter case; indeed, when 
$n \bmod 12=4$, we have $
\V{10a} = \V{10b} = \V{10c}$ and $ 
\V{11a} = \V{11b} = \V{11c}.
$ The two polytopes (with $n\bmod 12=2$ or 4) share six facets, namely,
\begin{align*}
m_{12} & \ge 0,\\
m_{13} & \ge 0,\\
m_{33} & \ge 0,\\
-4 m_{12} - 4 m_{13} + m_{33} & \ge -n - 5,\\
(n-4) m_{12} + \left(\frac{n}{2} - 2\right) m_{13} + \left(3 - \frac{n}{2}\right) m_{33} & \ge 2n - 8,\\
(n-7) m_{12} + (n-6) m_{13} + (4 - n) m_{33} & \ge 2n - 14.
\end{align*}
Moreover, when $n\bmod 12=2$, $\mathcal{P}_{n,n-1}$ possesses the following four additional facets.
\begin{align*}
-2 m_{12} -  m_{13}  & \ge \frac{-2(n+1)}{3},\\
-m_{12} + m_{33} & \ge \frac{-n-2}{4},\\
-2 m_{12} -2 m_{13} +  m_{33} & \ge \frac{-2 (n+4)}{3},\\
(2n-16) m_{12} +(2n-13) m_{13} + (10-2n) m_{33} & \ge 4n-32,
\end{align*} 
whereas for 
$n \bmod 12=4$, $\mathcal{P}_{n,n-1}$ possesses the following additional facet:
\begin{align*}
-3 m_{12} -2 m_{13} +  m_{33} & \ge \frac{-3 (n+4)}{4}.
\end{align*}

As reported in~\cite{delta3}, a total of 96 polytopes exist, each characterized by its own distinct set of facets. An online tool, called ChemicHull \cite{chemichull} has been developed, that allows the visualization of these 96 polytopes.
When neither 
$n$ nor 
$m$ is specified, ChemicHull returns the complete set of the 96 polytopes. It is possible to filter these polytopes by assigning specific values to 
$n$ and/or 
$m$, or by using formulas involving these parameters. For example, leaving 
$n$ free while specifying 
$m=n-1$ yields the twelve polytopes corresponding to chemical trees. 
Once 
$n$
and 
$m$ are specified, ChemicHull provides the V-Representation and the H-representation of the corresponding polytope. As an illustration, Figure \ref{fig:chem1} shows the 7 extreme points and the 8 facets of 
$\mathcal{P}_{8,8}$.
	 It should be noted that, unlike in this paper, the coordinates are provided directly in terms of the five values 
$m_{12}$, $m_{13}$, $m_{22}$, $m_{23}$, and $m_{33}$ (in this order), eliminating the need for the user to calculate 
$m_{22}$ and 
$m_{23}$.

\begin{figure}[h!]
\begin{center}
\includegraphics[scale=0.57]{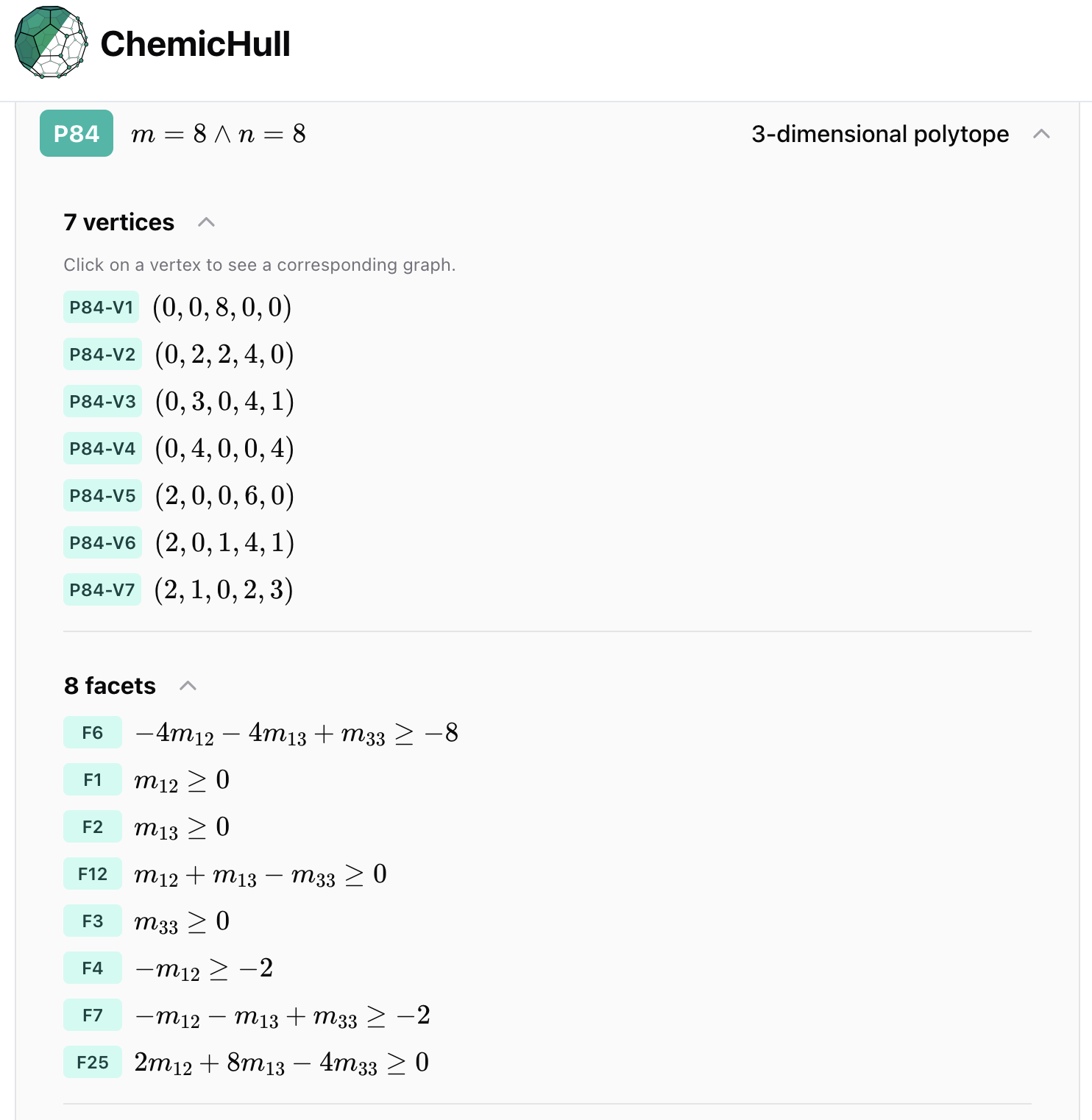}
\caption{V-Representation and H-representation of  $\mathcal{P}_{8, 8}$ in ChemicHull.}\label{fig:chem1}
\end{center}
\end{figure}

In addition to listing the extreme points and facets, the interface provides a three-dimensional graphical representation of the polytopes (see Figure~\ref{fig:chem2}).

\begin{figure}[h!]
\begin{center}
\includegraphics[scale=0.57]{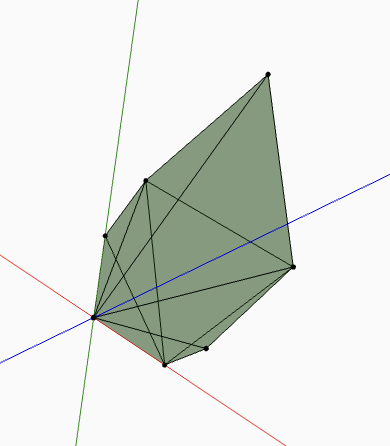}
\caption{A three-dimensional graphical representation of the polytope $\mathcal{P}_{8, 8}$}\label{fig:chem2}
\end{center}
\end{figure}
The interface also allows users to specify a degree-based topological index with ease. For example, Figure~\ref{fig:chem3} demonstrates how to define the Randić index.
 \begin{figure}[h!]
\begin{center}
\includegraphics[scale=0.43]{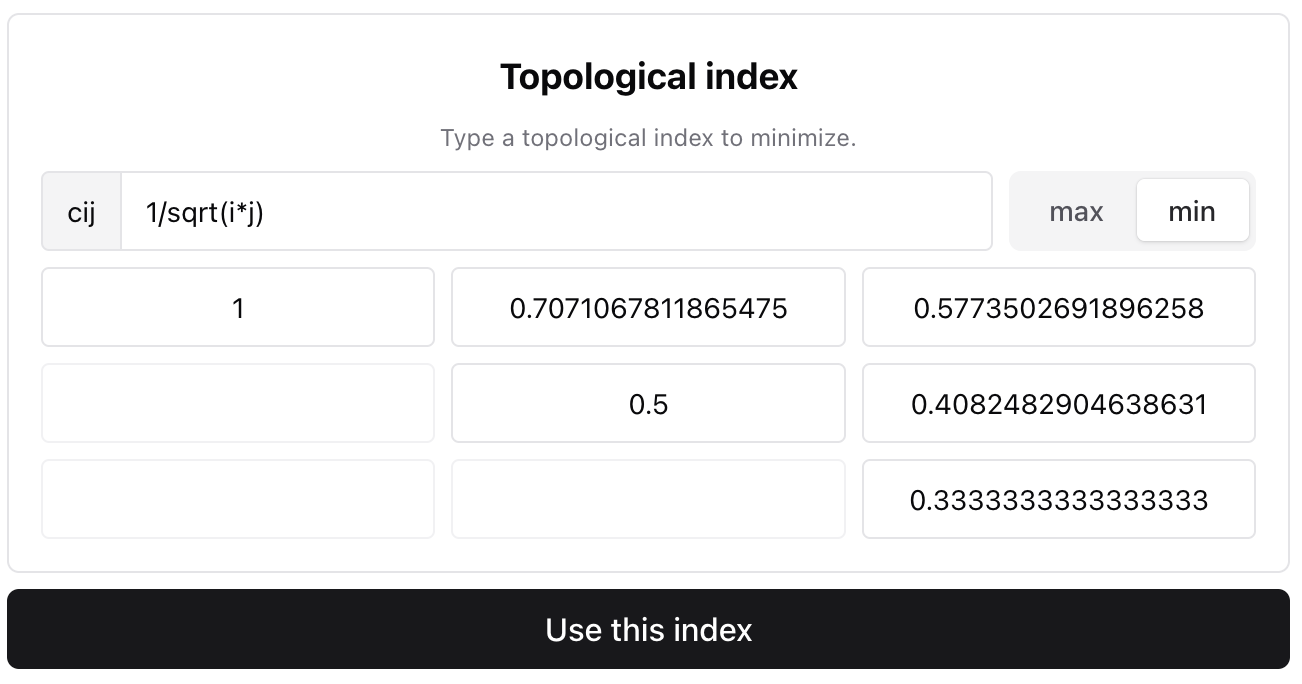}
\caption{Definition of a degree-based topological index (in this case, the Randić index).}\label{fig:chem3}
\end{center}
\end{figure}
Once a degree-based topological index is defined, ChemicHull can display the coordinates of the chemical graphs that maximize or minimize this index (see Figure~\ref{fig:chem4}).
\begin{figure}[h!]
\begin{center}
\includegraphics[scale=0.52]{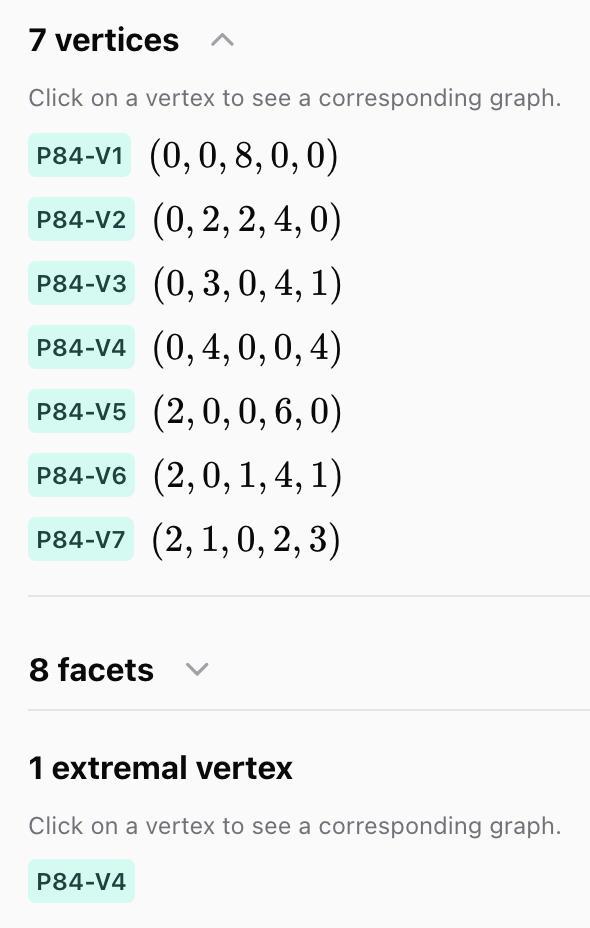}
\caption{Extreme point (P84-V4) that minimizes the Randić index in $\mathcal{P}_{8, 8}$.}\label{fig:chem4}
\end{center}
\end{figure}

Furthermore, users can visualize any graph corresponding to a selected point by simply clicking on it (see Figure~\ref{fig:chem5}). The construction of a graph from the coordinates 
$m_{ij}$ follows the algorithm described in~\cite{Hansen2017}.

\begin{figure}[h!]
\begin{center}
\includegraphics[scale=0.19]{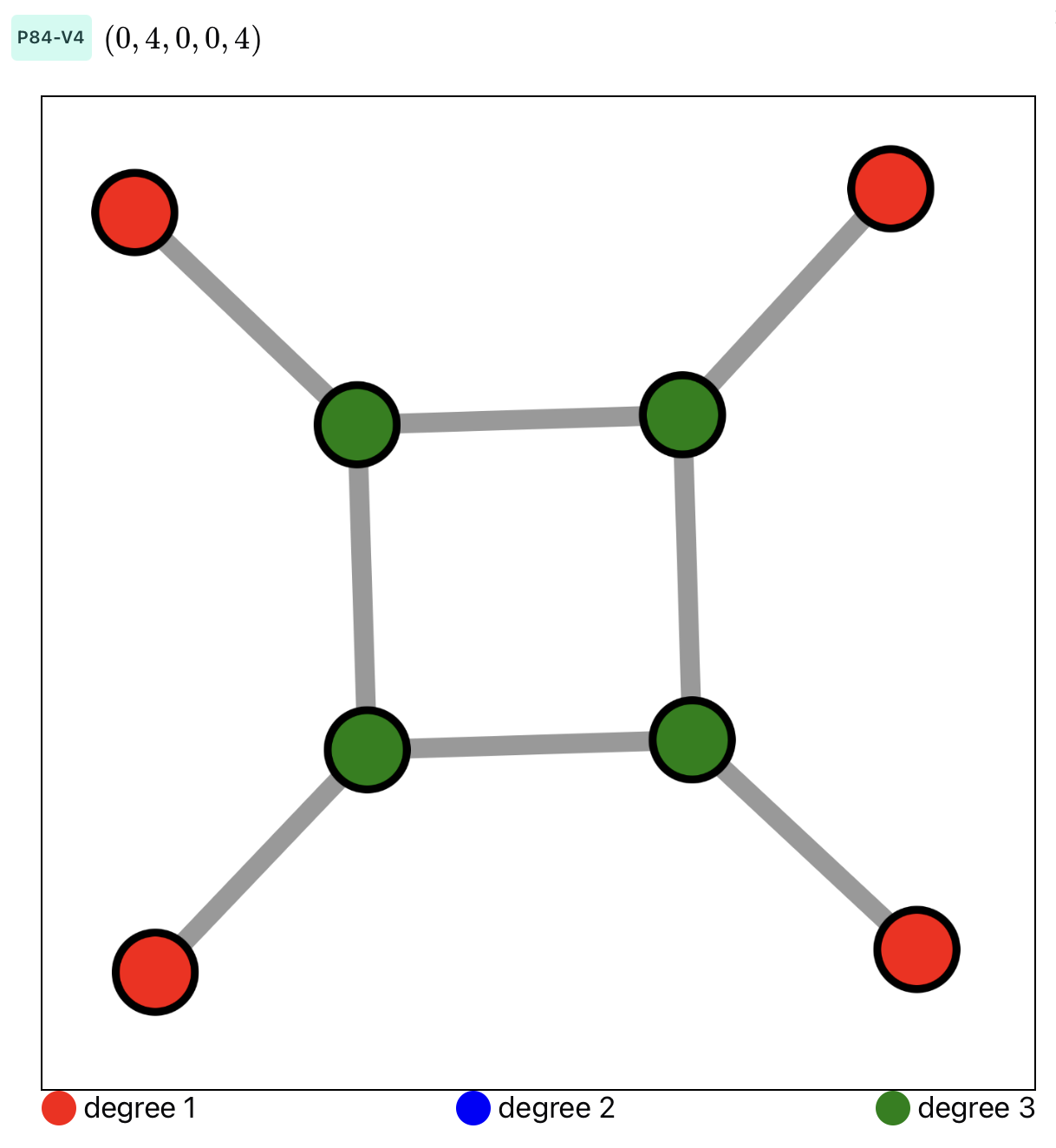}
\caption{A graph corresponding to the extreme point P84-V4=$(0,4,0,0,4)$ in $\mathcal{P}_{8, 8}$.}\label{fig:chem5}
\end{center}
\end{figure}

\section{Examples illustrating the applications of ChemicHull}\label{sec:application}

In this section, we reproduce established results from the literature to demonstrate the effectiveness of the approach based on the extreme points of the polytopes $\mathcal{P}_{n,m}$.
We  focus on the following degree-based topological indices:
\begin{itemize}[nosep]
    \item the generalized Randi\'c index $R_{\alpha}$ is  defined by $c_{ij}=(ij)^\alpha$ \cite{li2006}. When $\alpha=-\frac{1}{2}$, this is the standard Randi\'c index introduced in 1975 \cite{Randic1975};
    \item as mentioned in \cite{Gutman2021}, the Sombor index, denoted $SO$, is defined by $c_{ij}=\sqrt{i^2+j^2}$, while the reduced Sombor index, denoted $SO_{red}$, is defined by $c_{ij}=\sqrt{(i-1)^2+(j-1)^2}$;
     \item the atom-bond sum-connectivity index, denoted $ABS$ and introduced in \cite{ABS} is defined by $c_{ij}=\sqrt{\frac{i+j-2}{i+j}}$.\\   
\end{itemize}

Recall that for a degree-based topological index $I$, we have $\widehat{I}(m_{12},m_{13},m_{33})=c'_{12}m_{12}+c'_{13}m_{13}+c'_{33}m_{33}$, where 
\begin{itemize}[nosep]
    \item $c'_{12}=c_{12}-4c_{22}+3c_{23}$,
\item $c'_{13}=c_{13}-3c_{22}+2c_{23}$,
\item $c'_{33}=c_{22}-2c_{23}+c_{33}.$\\
\end{itemize}
The $c'_{ij}$ values for the above topological indices are given in Table \ref{tab:c'ij}.
The examples presented in the following subsections cover a broad spectrum of cases, as the optimal graphs in $\mathcal{G}_3(n,m)$ correspond to $\V{1}$, $\V{2}$, $\V{3}$, $\V{11c}$ or $\V{12b}$, depending on the chosen topological index and the condition relating $m$ to $n$. The reader interested in more examples is invited to consult \cite{article33}, in which the extremal graphs in $\mathcal{G}_3(n,m)$ for 33 topological indices are described.

\begin{table}[!htb]
	\centering
	\caption{Some values associated with five topological indices}\label{tab:c'ij}~\\[1ex]
	\begin{tabular}{c|l|c|l|l|l|} 
    \multicolumn{1}{c}{}&\multicolumn{1}{c}{$R_{-\frac{1}{2}}$}&\multicolumn{1}{c}{$R_{-1}$}&\multicolumn{1}{c}{$SO$ }&\multicolumn{1}{c}{$SO_{red}$}&\multicolumn{1}{c}{$ABS$}\\
		\cline{2-6}
	$c_{ij}$& \multicolumn{1}{c|}{$\frac{1}{\sqrt{ij}}$}&$\frac{1}{ij}$&$\sqrt{i^2+j^2}$&$\sqrt{(i-1)^2+(j-1)^2}$&$\sqrt{\frac{i+j-2}{i+j}}$\\ 		\cline{2-6}
    $c'_{12}$&$\simeq -0.068$&0&$\simeq 1.739$&$\quad\quad\quad\simeq 2.051$&$\simeq 0.073$\\\cline{2-6}
    $c'_{13}$&$\simeq -0.106$&$-\frac{1}{12}$&$\simeq 1.888$&$\quad\quad\quad\simeq 2.229$&$\simeq 0.135$\\\cline{2-6}
    $c'_{33}$&$\simeq 0.017$&$\frac{1}{36}$&$\simeq -0.140$&$\quad\quad\quad\simeq -0.229$&$\simeq -0.026$\\\cline{2-6}
	\end{tabular}
\end{table}

\subsection{Conditions for paths to be optimal chemical trees}
\label{chemical trees}

The following results provide  sufficient conditions on a degree-based topological index for $\Path{n}$ to be an optimal graph in $\mathcal{G}_3(n,n-1)$. 
\begin{thm}\label{pathmax}
    If $I$ is a degree-based topological index such that $c'_{13}<c'_{12}<-c'_{33}<0$, then $\Path{n}$ is the only tree in $\mathcal{G}_3(n,n-1)$ that maximizes  $I$.
\end{thm}
\begin{proof}
Let $I$ be a degree-based topological index such that $c'_{13}<c'_{12}<-c'_{33}<0$. We have to prove that $\V{2}=(2,0,0)$ is the unique point that maximizes $\widehat{I}$. We know from Table \ref{tab:extr_points1} that the candidate points to maximize $\widehat{I}$ over all chemical trees in $\mathcal{G}_3(n,n-1)$ of order at least 13 are $\V{2}$, $\V{7a}$, $\V{7b}$, $\V{7c}$, $\V{8c}$, $\V{10a}$, $\V{10b}$, $\V{10c}$, $\V{11a}$, $\V{11b}$, $\V{11c}$ and $\V{12a}$ as well as $\V{8d}$ if $n$ is odd. Note that 
\begin{itemize}[nosep]
     \item  $\widehat{I}(\V{7a})$ and $\widehat{I}(\V{7c})$ are at most  $\lceil\frac{n-4}{4}\rceil c'_{12}<2c'_{12}$;
     \item $\widehat{I}(\V{7b})\leq\lceil\frac{n+5}{4}\rceil c'_{12}+3c'_{33}\leq 5c'_{12}+3c'_{33}<2c'_{12}$;
     \item   $\widehat{I}(\V{8c})$ and $\widehat{I}(\V{8d})$ are at most  $\lceil\frac{n-1}{2}\rceil c'_{13}+
     \lfloor\frac{n-4}{2}\rfloor c'_{33}\leq  (\lceil\frac{n-1}{2}\rceil -
     \lfloor\frac{n-4}{2}\rfloor ) c'_{13} = 2 c'_{13} <2c'_{12}$;
     \item   $\widehat{I}(\V{10a})$, $\widehat{I}(\V{10b})$ and $\widehat{I}(\V{10c})$ are at most  $\lceil\frac{n-3}{3}\rceil c'_{13}+2c'_{33}\leq 4c'_{13}+2c'_{33}<2c'_{13}<2c'_{12}$;
     \item   $\widehat{I}(\V{11a})$, $\widehat{I}(\V{11b})$ and $\widehat{I}(\V{11c})$ are at most  $\lceil\frac{n+2}{3}\rceil c'_{12}+
     \lfloor\frac{n-5}{3}\rfloor c'_{33} < (\lceil\frac{n+2}{3}\rceil -
     \lfloor\frac{n-5}{3}\rfloor )c'_{12}  = 3 c'_{12} <2c'_{12}$;
\end{itemize}
Since $\widehat{I}(\V{2})=2c'_{12}$, we conclude that for $n\geq 13$, the path is the only chemical tree in $\mathcal{G}_3(n,n-1)$ that maximizes the $I$ index.
For $3\leq n\leq 12$, it is easy to verify from Table \ref{tab:extr_points2} that (2,0,0) (that is, the path) maximizes $\widehat{I}$. 
\end{proof}

Given that  minimizing $\widehat{I}$ is equivalent to maximizing $-\widehat{I}$, we can derive the following corollary.
\begin{cor}\label{pathmin}
    If $I$ is a degree-based topological index such that $c'_{13}>c'_{12}>-c'_{33}>0$, then $\Path{n}$ is the only  tree in $\mathcal{G}_3(n,n-1)$ that minimizes  $I$.
\end{cor}

It is shown in \cite{CAPOROSSI1999469} and \cite{GUTMAN1999366} that  $\Path{n}$ maximizes $R_{-\frac{1}{2}}$ in $\mathcal{G}_4(n,n-1)$. Hence, this is also true for chemical trees of maximum degree at most 3. This is a direct corollary of Theorem \ref{pathmax}. Indeed, as shown in Table \ref{tab:c'ij}, for this index, we have 
$$c'_{13}\simeq -0.106<c'_{12}\simeq -0.068<-c'_{33}=-0.017<0.$$

As another example, the authors of \cite{Liu2022} have characterized the 14 trees in $\mathcal{G}_4(n,n-1)$ with smallest 
Sombor and reduced Sombor indices. In particular, they have shown that $\Path{n}$ has minimum value. Hence, this is also true for chemical trees of maximum degree at most 3, and this is a direct consequence of Corollary \ref{pathmin}. Indeed,  as shown in Table \ref{tab:c'ij}, for the $SO$  index, we have 
$$c'_{13}\simeq 1.888>c'_{12}\simeq 1.739>-c'_{33}\simeq0.140>0,$$
while for $SO_{red}$, we have 
$$c'_{13}\simeq 2.229>c'_{12}\simeq 2.051>-c'_{33}\simeq0.229>0.$$
As a further example, consider the $ABS$ index. The authors of \cite{Zuo2024} show, among other results, that $\Path{n}$ minimizes this index among all chemical trees in $\mathcal{G}_4(n,n-1)$.  Hence, this is also true for trees of maximum degree at most 3, and this follows directly from  Corollary \ref{pathmin}. Indeed, as shown in Table \ref{tab:c'ij}, for the $ABS$ index, we have 
$$c'_{13}\simeq 0.135>c'_{12}\simeq 0.073>-c'_{33}\simeq 0.026>0.$$

\subsection{Conditions for cycles to be optimal unicyclic chemical graphs}
\label{chemical trees}

The following results provide sufficient conditions on a degree-based topological index for $\C{n}$ to be an optimal graph in $\mathcal{G}_3(n,n)$. 
\begin{thm}\label{cyclemin}
    If $I$ is a degree-based topological index such that $\max\{0,-c'_{33}\}<\min\{c'_{12},c'_{13}\}$, then $\C{n}$ is the only graph in $\mathcal{G}_3(n,n)$ that minimizes  $I$.
\end{thm}
\begin{proof}
Let $I$ be a degree-based topological index such that $\max\{0,-c'_{33}\}<\min\{c'_{12},c'_{13}\}$. We have to prove that $\V{1}=(0,0,0)$ is the unique point that minimizes $\widehat{I}$. We know from Table \ref{tab:extr_points1} that the candidate points to minimize $\widehat{I}$ over all graphs in $\mathcal{G}_3(n,n)$ of order at least 12 are $\V{1}$, $\V{7a}$, $\V{7b}$, $\V{7c}$, $\V{8c}$, $\V{10a}$, $\V{10b}$, $\V{10c}$, $\V{11a}$, $\V{11b}$ and $\V{11c}$ as well as $\V{8d}$ if $n$ is odd. It is easy to observe that $\widehat{I}($V$i)$ is strictly positive for all these points, except for $\V{1}$ for which $\widehat{I}(\V{1})=0$. For smaller values of $n$, one can check from Table \ref{tab:extr_points2} that $(0,0,0)$ is always the point with minimum $I$ index value. We therefore conclude that for $n\geq 3$, $\C{n}$ is the only graph in $\mathcal{G}_3(n,n)$ that minimizes the $I$ index.
\end{proof}

Given that  maximizing $\widehat{I}$ is equivalent to minimizing $-\widehat{I}$,  this leads to the following corollary.

\begin{cor}\label{cyclemax}
    If $I$ is a degree-based topological index such that $\max\{c'_{12},c'_{13}\}<\min\{0,-c'_{33}\}$, then $\C{n}$ is the only chemical graph in $\mathcal{G}_3(n,n)$ that maximizes  $I$.
\end{cor}

It is shown in \cite{CAPOROSSI200385} that $\C{n}$ maximizes $R_{-\frac{1}{2}}$ over all graphs in $\mathcal{G}_4(n,n)$. Hence, this is also true for the graphs in $\mathcal{G}_3(n,n)$. This follows from  Corollary \ref{cyclemax}. Indeed, as shown in Table \ref{tab:c'ij}, for the $R_{-\frac{1}{2}}$ index, we have 
$$\max\{c'_{12},c'_{13}\}=c'_{12}\simeq -0.068<-c'_{33}\simeq -0.017<0.$$

As a second example, it is shown in  \cite{Zuo2024} that $\C{n}$ minimizes the $ABS$ index in $\mathcal{G}_4(n,n)$.  Hence, this is also true for graphs in $\mathcal{G}_3(n,n)$, and this is a direct consequence of Theorem \ref{cyclemin}. Indeed, for this index, we have 
$$\min\{c'_{12},c'_{13}\}=c'_{12}\simeq 0.073>-c'_{33}\simeq 0.026>0.$$

\subsection{Trees that maximize $R_{-1}$ in $\mathcal{G}_3(n,n-1)$}
\label{sub:appliRautenbach}

The problem of maximizing $R_{-1}$ over all graphs in $\mathcal{G}_3(n,n-1)$ is solved in~\cite{Rautenbach2003}. We can easily prove the same result with our polyhedral approach. Indeed, 
assuming $m=n-1$, we get from Table~\ref{tab:extr_points2} that the points that maximize $\widehat{R}_{-1}$ for $3\leq n\leq 12$ are 
\begin{itemize}
    \item $(2,0,0)=\V{2}$ if $n\in\{3,4,5,6\}$,
    \item $(2,0,0)=\V{2}$ and $(3,0,0)=\V{11c}$ if $n\in\{7, 8,9\}$, and
        \item $(4,0,1)=\V{11c}$ if $n\in\{10,11,12\}$.
\end{itemize}

For $n\geq 13$, we know from Table~\ref{tab:extr_points1} that the only candidate points to maximize $\widehat{R}_{-1}$ for graphs in $\mathcal{G}_3(n,n-1)$  are $\V{2}$, $\V{7a}$, $\V{7b}$, $\V{7c}$, $\V{8c}$, $\V{10a}$, $\V{10b}$, $\V{10c}$, $\V{11a}$, $\V{11b}$, $\V{11c}$ and $\V{12a}$ as well as $\V{8d}$ if $n$ is odd. Clearly, the only candidates V$i$ for which $\widehat{R}_{-1}($V$i)$ can be strictly positive are $\V{7b}$, $\V{11a}$, $\V{11b}$ and $\V{11c}$, and it is easy to observe that
\begin{itemize}[nosep]
    \item $\widehat{R}_{-1}(\V{11a})=\widehat{R}_{-1}(\V{11b})=\widehat{R}_{-1}(\V{11c})$ when $n\mod 3=1$ and both $\widehat{R}_{-1}(\V{11a})$ and $\widehat{R}_{-1}(\V{11b})$ are strictly smaller than $\widehat{R}_{-1}(\V{11c})$ otherwise;
    \item $\widehat{R}_{-1}(\V{7b})=\widehat{R}_{-1}(\V{11c})$ if $n=13$ or 16, and $\widehat{R}_{-1}(\V{7b})<\widehat{R}_{-1}(\V{11c})$ otherwise.
\end{itemize}

Hence, in all cases, $\widehat{R}_{-1}($V$i)$ is maximized with $\V{11c}$ when $n\geq 13$.
In summary, $\V{2}$ maximizes $\widehat{R}_{-1}$ if $n<7$ while $\V{11c}$ maximizes $\widehat{R}_{-1}$ if $n\geq 7$.

The value $R_{-1}(V)$ of the $R_{-1}$ index of a point $V$ corresponding to a chemical tree of order $n$  is $R_{-1}(V)=\widehat{R}_{-1}(V)+C_{R_{-1}}(n,m)$ where
$$C_{R_{-1}}(n,m)=(6n-5m)c_{22}+(6m-6n)c_{23}=\frac{n+5}{4}-1=\frac{n+1}{4}.$$

Since $R_{-1}(\V{2})=\frac{n+1}{4}$ and $$R_{-1}(\V{11c})=\frac{1}{36}\left\lfloor\frac{n-7}{3}\right\rfloor+\frac{n+1}{4}=\frac{7n}{27}+\frac{20-(n+2)\bmod 3}{108},$$
we can conclude that given any chemical tree $T$ in $\mathcal{G}_3(n,n-1)$, as shown in \cite{Rautenbach2003}, we have 

$$
R_{-1}(T) \leq \left\{ 
\begin{array}{ll}
\frac{n+1}{4} & \mbox{if } n\in\{3,4,5,6\},\\
\frac{7n}{27}+\frac{1}{6} & \mbox{if }n\geq 7\mbox{ and }n\equiv 0\bmod 3,\\
\frac{7n}{27}+\frac{5}{27} & \mbox{if }n\geq 7\mbox{ and }n\equiv 1\bmod 3,\\
\frac{7n}{27}+\frac{19}{108} & \mbox{if }n\geq 7\mbox{ and }n\equiv 2\bmod 3.
\end{array}\right.$$

\subsection{More optimal values for the  $R_{-\frac{1}{2}}$ and $ABS$ indices}

Let's consider bicyclic graphs. It is proved in \cite{CAPOROSSI200385} that over all graphs in $\mathcal{G}_4(n,n+1)$, $\V{3}$ maximizes $R_{-\frac{1}{2}}$, and in \cite{Zuo2024} that the same point minimizes the $ABS$ index. Hence, this is also true when the maximum degree is at most 3, and this can be proved with our polyhedral approach. Indeed, 
for bicyclic chemical graphs with at least 11 vertices, Table~\ref{tab:extr_points1} shows  that the candidate points to achieve an optimal value are $\V{1}$, $\V{3}$, $\V{7a}$, $\V{7b}$, $\V{7c}$, $\V{8c}$, $\V{9a}$, $\V{9b}$, $\V{9c}$, $\V{10a}$, $\V{10b}$, $\V{10c}$, $\V{11a}$, $\V{11b}$ and $\V{11c}$ as well as $\V{8d}$ if $n$ is odd. \begin{itemize}
    \item 
It is easy to observe that $\widehat{R}_{-\frac{1}{2}}($V$i)$ is less than or equal to zero for all these points, except for $\V{3}$ for which $\widehat{R}_{-\frac{1}{2}}(\V{3})\simeq 0.017$. For smaller values of $n$, it is easy to check from Table~\ref{tab:extr_points2} that $(0,0,1)$ is always the point with maximum $R_{-\frac{1}{2}}$ index.
    \item 
It is easy to observe that $\widehat{ABS}($V$i)$ is non-negative for all these points, except for $\V{3}$ for which $\widehat{ABS}(\V{3})\simeq -0.026$. For smaller values of $n$, it is easy to check from Table~\ref{tab:extr_points2} that $(0,0,1)$ is always the point with minimum $ABS$  index.\end{itemize}

Let's now consider tricyclic graphs. It is argued in \cite{CAPOROSSI200385} (without proof) that the graph that maximizes $R_{-\frac{1}{2}}$ over all graphs in $\mathcal{G}_4(n,n+2)$ has 4 edges with both endpoints of degree 3 and no vertex of degree 1 or 4. This corresponds to the point $(0,0,4)$ in $\mathcal{G}_3(n,n+2)$. This is however wrong. Indeed, we prove here below that the tricyclic graph $\ensuremath{{\sf T}}_n$ (i.e., $\V{12b}=(0,0,5)$) is the unique graph in $\mathcal{G}_3(n,n+2)$ that maximizes $R_{-\frac{1}{2}}$ in $\mathcal{G}_3(n,n+2)$:
\begin{itemize}
\item For graphs in $\mathcal{G}_3(n,n+2)$ with at least 11 vertices, Table~\ref{tab:extr_points1} shows  that the candidate points to achieve the maximum value are $\V{1}$, $\V{7b}$, $\V{7c}$, $\V{8c}$, $\V{10a}$, $\V{10c}$, $\V{11a}$, $\V{11b}$, $\V{11c}$, $\V{12b}$ and $\V{12c}$, as well as $\V{7a}$ and $\V{10b}$ if $n\notin \{11,12,15\}$ and $\V{8d}$ if $n$ is odd (since $n+1<n+2<\frac{6n}{5}$).  It is easy to observe that $\widehat{R}_{-\frac{1}{2}}($V$i)$ is non-positive for all these points, except $\widehat{R}_{-\frac{1}{2}}(\V{12c})=\widehat{R}_{-\frac{1}{2}}(1,0,7)=c'_{12}+7c'_{33}\simeq 0.05$ and $\widehat{R}_{-\frac{1}{2}}(\V{12b})=\widehat{R}_{-\frac{1}{2}}(0,0,5)\simeq 0.08$.
\item For $n=10$, we have $m=n+2=\frac{6n}{5}$, and Table~\ref{tab:extr_points1} shows that the candidate points are 
\V{6}$, \V{8c}$, $\V{11a}$, $\V{11b}$, $\V{11c}$, $\V{12b}$ and $\V{12c}$, as well as $\V{8d}$ if $n$ is odd. Since $\V{6}=\V{1}=(0,0,0)$, we have observed in the previous case that $\V{12b}=(0,0,5)$ maximizes $R_{-\frac{1}{2}}$.
\item For $5\leq n\leq 9$, it is easy to check from Table \ref{tab:extr_points2} that $(0,0,5)$ is also always the point with maximum $R_{-\frac{1}{2}}$ index. 
\end{itemize}
We therefore conclude that for $n\geq 5$, $\V{12b}$ maximizes the $R_{-\frac{1}{2}}$ index over all tricyclic chemical graphs. \newline

Let's conclude this section with the graphs in $\mathcal{G}_3(n,n+2)$ that minimize the $ABS$ index. It is proved in \cite{Zuo2024} that for a maximum degree at most 4, the optimal tricyclic graph is $\V{12b}$. Hence this is also true when the maximum degree is at most 3 which we now prove with our polyhedral approach. Recall that the candidate points to achieve the minimum value are $\V{1}$, $\V{7b}$, $\V{7c}$, $\V{8c}$, $\V{10a}$, $\V{10c}$, $\V{11a}$, $\V{11b}$, $\V{11c}$, $\V{12b}$ and $\V{12c}$, as well as $\V{7a}$ and $\V{10b}$ if $n\notin \{11,12,15\}$ and $\V{8d}$ if $n$ is odd.
Note that $\widehat{ABS}(\V{12b})\simeq -0.128$ and $\widehat{ABS}(\V{12c})\simeq -0.106$, and it is easy to observe that $\widehat{ABS}($V$i)$ is non-negative for all the other points, except in the following cases:
\begin{itemize}[nosep]
    \item if $n=10$, $\widehat{ABS}(\V{11b}){=}\widehat{ABS}(\V{11c}){=}\widehat{ABS}(2,0,8)=2c'_{12}+8c'_{33}>5c'_{33}$;
    \item if $n=11$, $\widehat{ABS}(\V{7b})=\widehat{ABS}(1,0,3)=c'_{12}+3c'_{33}>5c'_{33}$, $\widehat{ABS}(\V{11a})=\widehat{ABS}(2,0,7)=2c'_{12}+7c'_{33}>5c'_{33}$,  and $\widehat{ABS}(\V{11c})=\widehat{ABS}(2,0,8)=2c'_{12}+8c'_{33}>5c'_{33}$;
    \item if $n=12$, $\widehat{ABS}(\V{11a}){=}\widehat{ABS}(2,0,6)=2c'_{12}+6c'_{33}>5c'_{33}$, and $\widehat{ABS}(\V{11c}){=}\widehat{ABS}(2,0,8) = 2c'_{12}+8c'_{33}>5c'_{33}$;
        \item if $n=13$, $\widehat{ABS}(\V{11a}) {=}\widehat{ABS}(\V{11b}){=}\widehat{ABS}(\V{11c}){=}\widehat{ABS}(3,0,9)=3c'_{12}+9c'_{33}>5c'_{33}$;
    \item if $n=14$ or 15, $\widehat{ABS}(\V{11c}){=}\widehat{ABS}(3,0,9)=3c'_{12}+9c'_{33}>5c'_{33}$.
\end{itemize}
We deduce that for $n\geq 10$, $\V{12b}$ minimizes the $ABS$ index over all graphs in $\mathcal{G}_3(n,n+2)$. 
For $5\leq n\leq 9$, it is easy to check from Table \ref{tab:extr_points2} that $(0,0,5)$ is also always the point with minimum $ABS$ index.

\section{Concluding remarks and future work}\label{sec:futurework}

In this study, we have shown that the polyhedral approach offers a powerful and systematic framework for identifying families of chemical graphs that are extremal with respect to degree-based topological indices. Beyond reproducing many established results from the literature, this method also allows for their verification and, when required, correction, as exemplified by our discovery of a counterexample to a previously reported extremal result involving the Randić index.

We have also introduced ChemicHull, an online tool that enables researchers to easily explore and apply the results presented in this work. The tool offers an intuitive interface for visualizing extremal chemical graphs and analyzing optimal values of degree-based topological indices.

Moreover, our findings indicate that some graphs can never serve as extreme points under any degree-based topological index. Therefore, if a molecule of interest in extremal chemistry is not represented among the extremal points identified in this study, it must instead be described by a topological index beyond the degree-based class.

An important direction for future work is to extend these results to chemical graphs with a maximum degree of at most 4. Preliminary investigations, however, suggest that such an extension will be considerably more complex, both combinatorially and computationally.

\bibliographystyle{acm}

\end{document}